\documentclass[11pt,reqno]{amsart}
\usepackage{amsmath}
\usepackage{amssymb}

\def\squarebox#1{\hbox to #1{\hfill\vbox to #1{\vfill}}}

\theoremstyle{plain}

\newtheorem{thm}{Theorem}
   
\newtheorem{cor}{Corollary}
   
\newtheorem{lem}{Lemma}

\newcommand{\supp}{\operatorname{supp}}

\renewcommand{\Re}{\mathop{\rm Re}\nolimits}
\renewcommand{\Im}{\mathop{\rm Im}\nolimits}

\def\la{\langle}
\def\ra{\rangle}
\def\lx{\la x \ra}
\def\ii{{\bf i}}

\newtheorem{rem}{Remark}
\newtheorem{prop}{Proposition}

\numberwithin{equation}{section}

\begin{document}
\def\C{{\mathbb C}}
\def\R{{\mathbb R}}
\def\N{{\mathbb N}}
\def\Z{{\mathbb Z}}
\def\T{{\mathbb T}}
\def\Q{{\mathbb Q}}
\def\SP{{\mathbb S}}
\def\d{{\partial}}
\def\mc{{\mathcal H}}
\def\1b{{\mathbb I}}
\def\tr{{\rm tr}\:}
\def\pv{\partial_x V}

\title[Representation of the spectral shift function]
{ Spectral shift function for operators with crossed magnetic and
electric fields}
\author[M. Dimassi]{Mouez Dimassi}
\author[V. Petkov]{Vesselin Petkov}

\address {D\'epartement de Math\'ematiques,
Universit\'e Paris 13, France, Avenue J.-B. Cl\'ement, 93430 Villetaneuse, France}
\email{dimassi@math.univ-paris13.fr}
\address {Universit\'e Bordeaux I, Institut de Math\'ematiques de Bordeaux,  351, Cours de la Lib\'eration, 33405  Talence, France}
\email{petkov@math.u-bordeaux1.fr}
\thanks{The second author was partially supported by the ANR project NONAa}
\maketitle
\begin{abstract} We obtain a representation formula for the derivative of the spectral shift function $\xi(\lambda; B, \epsilon)$ related to the operators $H_0(B,\epsilon) = (D_x - By)^2 + D_y^2 + \epsilon  x$ and $H(B, \epsilon) = H_0(B, \epsilon) + V(x,y), \: B > 0, \epsilon > 0$. We establish a limiting absorption principle for $H(B, \epsilon)$ and an estimate ${\mathcal O}(\epsilon^{n-2})$ for $\xi'(\lambda; B, \epsilon)$, provided $\lambda \notin \sigma(Q)$, where $Q = (D_x - By)^2 + D_y^2 + V(x,y).$

\end{abstract}

\section{Introduction}
\renewcommand{\theequation}{\arabic{section}.\arabic{equation}}
\setcounter{equation}{0}

Consider the two-dimensional Schr\"odinger operator with
homogeneous magnetic and electric fields

$$H = H(B, \epsilon) = H_0(B, \epsilon) + V(x, y),\: D_x = -\ii \partial_x,\: D_y = -\ii \partial_y,$$
 where
$$H_0 = H_0(B, \epsilon) = (D_x - By)^2 + D_y^2 + \epsilon x.$$
Here $B > 0$ and $\epsilon > 0$ are proportional to the strength
of the homogeneous magnetic and electric fields. We assume that $V, \partial_x V \in C^0(\R^2; \R) \cap
L^\infty(\R^2; \R))$ and $V(x, y)$ satisfies the estimate
\begin{equation}\label{eq:1.1}
| V(x, y)|  \leq C(1 +
|x|)^{-2 -\delta }(1 + |y|)^{-1- \delta },
\delta > 0.
\end{equation}
For $\epsilon \not= 0$ we have $\sigma_{\rm ess}(H_0(B, \epsilon))
= \sigma_{\rm ess} (H(B, \epsilon)) = \R$. On the other hand, for decreasing potentials $V$ we may have embedded eigenvalues $\lambda \in \R$
and this situation is completely different from that with $\epsilon =
0$ when the spectrum of $H(B, 0)$ is formed by eigenvalues with
finite multiplicities which may accumulate only to Landau levels
$\lambda_n = (2n +1)B,\: n \in \N$ (see \cite{I},
\cite{MR}, \cite{RW} and the references cited there). The spectral
properties of $H$ and the existence of resonances have been
studied in \cite{FK1}, \cite{FK2}, \cite{DP2} under the assumption
that $V(x,y)$ admits a holomorphic extension in the $x$- variable
into a domain
$$\Gamma_{\delta_0} = \{z \in \C:\: 0 \leq |\Im z| \leq \delta_0\}.$$
Moreover, without any assumption on the analyticity of $V(x,y)$  we show in Proposition 2 below that the operator $(H - z)^{-1} - (H_0 - z)^{-1}$ for $z \in \C,\: \Im z \neq 0,$ is trace class and following the general setup \cite{K}, \cite{Ya},  we define the spectral shift function $\xi(\lambda)=
\xi(\lambda; B, \epsilon)$ related to $H_0(B, \epsilon)$ and $H(B,
\epsilon)$ by
$$\la \xi', f \ra = {\rm tr} \Bigl(f(H) - f(H_0)\Bigr),\: f \in C_0^{\infty}(\R).$$
By this formula $\xi(\lambda)$ is defined modulo a constant but for the analysis of the derivative $\xi'(\lambda)$ this is not important. Moreover, the above property of the resolvents and Birman-Kuroda theorem imply $\sigma_{\rm ac}(H_0(B, \epsilon)) = \sigma_{\rm ac}(H(B, \epsilon)) = \R.$
 A representation of the derivative $\xi'(\lambda; B, \epsilon)$ has been obtained in \cite{DP2} for strong magnetic fields $B \to +\infty$ under the  assumption that $V(x,y)$ admits an analytic continuation in $x$-direction. Moreover, the distribution of the resonances $z_j$ of the perturbed operator $H(B, \epsilon)$ has been examined in \cite{DP2} and a Breit-Wigner representation  of $\xi'(\lambda; B, \epsilon)$ involving the resonances $z_j$ was established.\\

In the literature there are a lot of works concerning 
Schr\"odinger operators with magnetic fields ($\epsilon = 0$) but there are only few ones dealing with magnetic
and Stark potentials ($\epsilon \not= 0$) (see \cite{FK1}, \cite{FK2}, \cite{DP2} and the
references given there). It should be mentioned that the tools in \cite{FK1}, \cite{FK2} and \cite{DP2} are
related to the resonances of the perturbed problem and to define the
resonances one supposes that the potential
$V(x,y)$ has an analytic continuation in $x$ variable.   In this paper we consider the operator $H$ without {\it any
assumption} on the analytic continuation of $V(x,y)$ and without
the {\it restriction} $B \to +\infty$. Our purpose is to study
 $\xi'(\lambda; B, \epsilon)$ and the
existence of embedded eigenvalues of $H$. To examine the behavior of the spectral shift function we need a representation of the derivative $\xi'(\lambda; B, \epsilon)$. The key point in this direction is the following

\begin{thm} Let $V, \partial_x V \in C^0(\R^2; \R) \cap L^{\infty}(\R^2; \R)$ and let $(1.1)$ hold for $V$ and $\partial_xV$.
 Then for every $f \in C_0^{\infty}(\R)$ and $\epsilon \neq 0$ we have
\begin{equation} \label{eq:1.2}
\tr\: \Bigl(f(H) - f(H_0)\Bigr) = - \frac{1}{\epsilon} \tr\:
\Bigl(\d_x V f(H)\Bigr).
\end{equation}
\end{thm}
The formula (\ref{eq:1.2}) has been proved by D. Robert and X.P.Wang
\cite{RW2} for Stark Hamiltonians in absence of magnetic field ($B = 0$). In fact, the
result in \cite{RW2} says that
\begin{equation} \label{eq:1.3}
\xi'(\lambda; 0, \epsilon) = -\frac{1}{\epsilon}\int_{\R^2} \pv
\frac{\partial e}{\partial \lambda}(x,y, x, y; \lambda, 0,
\epsilon) dxdy,
\end{equation}
where $e(., .; \lambda, 0, \epsilon)$ is the spectral function of
$H(0, \epsilon).$
 The presence of magnetic filed $B \not= 0$ and Stark potential lead to some serious difficulties. The operator $H$ is not elliptic for $|x| + |y| \to \infty$ and we have double characteristics. On the other hand, the commutator $[H, x]$ involves the term $(D_x - By)$ and it creates additional difficulties. The proof of Theorem 1 is long and technical. We are going to study the trace class properties of the operators  $\psi (H \pm \ii)^{-N}$, \,\,\,$\partial_x \circ\psi (H\pm \ii)^{-N-1}$, \,\,\,\,$(H\pm \ii)\partial_x\circ\psi  (H\pm \ii)^{-N-2}$ etc.
for $N \geq 2$ and $\psi \in C_0^{\infty}(\R^2)$ (see Lemmas 1 and
2). Moreover, by an argument similar to that in Proposition 2.1 in
\cite{DP2}, we obtain estimates for the trace norms of the
operators
$$(z - H)^{-1}V(z' - H)^{-1},\: V(z - H)^{-1}(z' - H)^{-1},\:z \notin \R, z' \notin \R$$
and we apply an approximation argument. Notice that in \cite{RW2} the spectral shift function is related to the trace of the {\it time delay} operator $T(\lambda)$ defined via the corresponding scattering matrix $S(\lambda)$ (see \cite{RW1}). In contrast to \cite{RW2}, our proof  is direct and neither $T(\lambda)$ nor $S(\lambda)$ corresponding to the operator $H(B, \epsilon)$ are used.\\

The second question examined in this work is the existence of
embedded real eigenvalues and the limiting absorption principle
for $H$. In the physical literature one conjectures that for
$\epsilon \not= 0$ there are no  embedded eigenvalues. We
establish in Section 3 a weaker result saying that in any interval
$[a, b]$ we may have at most a finite number embedded eigenvalues
with finite multiplicities. Under the assumption for analytic
continuation of $V$ it was proved in \cite{FK1} that for some
finite interval $[\alpha(B, \epsilon), \beta(B, \epsilon)]$ there
are no resonances $z$ of $H(B, \epsilon)$ with $\Re z \notin
[\alpha(B, \epsilon), \beta(B, \epsilon)]$. Since the real
resonances $z$ coincide with the eigenvalues of $H(B, \epsilon)$,
we obtain some information for the embedded eigenvalues. On the
other hand, exploiting the analytic continuation and the resonances we proved in \cite{DP2} that for $B \to +\infty$ the reals parts $\Re z_j$ of the resonances $z_j$ lie outside some neighborhoods of the Landau levels. Thus the Landau levels play a role in the distribution of the resonances. It is known that the spectrum of the operator $Q = (D_x - By)^2 + D_y^2 + V(x, y)$ with decreasing potential $V$ is formed by eigenvalues (see \cite{I}, \cite{MR}, \cite{RW}). In this paper we establish a limiting absorption principle for $\lambda \notin \sigma(Q)$. In particular, we show that there are no embedded eigenvalues outside $\sigma(Q)$. This agrees with the result in \cite{DP2} obtained under the restrictions on the behavior of $V$ and $B \to +\infty$. On the other hand, the result of Proposition 3 and the estimates (\ref{eq:4.3}) have been established by X. P. Wang \cite{W} for Stark operators with $B = 0$.\\

Following the results in Section 4 and the representation of
$\xi'(\lambda; B, \epsilon)$ given in \cite{DP2}, it is natural to
expect that for $\lambda  \notin \sigma(Q)$ the derivative of the
spectral shift function $\xi'(\lambda; B, \epsilon)$ must be
bounded. In fact, we prove the following stronger result.

\begin{thm} Let the potential $V\in C^\infty(\R^2;\R)$ satisfy with some $\delta > 0$ and $n \in N,\: n \geq 2$ the estimates
\begin{equation} \label{eq:1.4}
|\partial_x^{\alpha} \partial_y^{\beta} V(x,y)| \leq C_{\alpha,
\beta} (1 + |x|)^{-n-\delta -|\alpha|} (1 + |y|)^{-2 -\delta -
|\beta|},\: \forall \alpha, \forall \beta.
\end{equation}
Then for $\lambda_0 \notin \sigma(Q)$ we have
\begin{equation} \label{eq:1.5}
\xi'(\lambda; B, \epsilon) = {\mathcal O}(\epsilon^{n -2})
\end{equation}
uniformly for $\lambda$ in a small neighborhood $\Xi \subset \R$
of $\lambda_0.$
\end{thm}
The estimate (\ref{eq:1.5}) has been obtained in \cite{RW2} in the case of absence of magnetic field $B = 0$ (for a Breit-Wigner formula see \cite{MRW}, \cite{DP1} for Stark Hamiltonians and \cite{DP2} for the operator $H(B, \epsilon)$). Our approach is quite different from that in  \cite{RW2}. Our proof is going without an application of a representation similar to (\ref{eq:1.3}) which leads  to complications connected with the behavior of the spectral function $e(., .; \lambda, B, \epsilon)$ corresponding to $H(B, \epsilon).$  The formula (\ref{eq:1.2}) plays a crucial role and our analysis is based on a complex analysis argument combined with a representation of $f(H)$ involving the almost analytic continuation of $f \in C_0^{\infty}(\R)$. In this direction our argument is similar to that developed in \cite{DP1} and \cite{DP2}.\\

The plan of this paper is as follows. In Sect. 2 we establish
Theorem 1. The embedded eigenvalues and Mourre estimates are
examined in Sect. 3. In Sect. 4 we prove Proposition 3 concerning
the limiting absorption principle for $H(B, \epsilon).$ Finally, in Sect. 5 we establish Theorem 2.\\

{\it Acknowledgement.} The authors  are grateful to the referees for their thorough and careful reading of the paper. Their remarks and suggestions lead to an improvement of the  first version of this paper.

\section{Representation of the spectral shift function}
\renewcommand{\theequation}{\arabic{section}.\arabic{equation}}
\setcounter{equation}{0}

 Throughout this work we will use the notations of \cite{DS} for symbols and pseudodifferential operators.
In particular, if $m: {\R^4} \rightarrow [0,+\infty[$ is an order
function (see \cite{DS}, Definition 7.4), we say that $a(z,\zeta)\in S^0(m)$ if for
every $\alpha\in\N^4$ there exists $C_{\alpha} > 0$ such that
$$|\partial_{z,\zeta}^\alpha a(z,\zeta)|\leq C_{\alpha} m(z,\zeta).$$
In the special case when $m = 1$, we will write $S^0$ instead of
$S^0(1)$. We will use the standard Weyl quantization  of symbols.
More precisely, if $p(z,\zeta)$, $(z,\zeta)\in \R^4$,  is a symbol
in $S^0(m)$, then $P^w(z,D_z)$ is the operator defined by
$$ P^w(z,D_z)u(z)=(2\pi)^{-2}\iint e^{i(z-z')\cdot \zeta} p\Bigl({z+z'\over 2},\zeta\Bigr) u(z') dz'd\zeta,\,\,\,
\hbox { for }\,\, u\in {\mathcal S}(\R^2).$$
We denote  by $P^w(z, hD_z)$ the semiclassical quantization obtained as above by quantizing $p(z,h\zeta)$.\\

Our goal in this section is to prove Theorem 1. For this purpose
we need some Lemmas. We set 
$$Q_0 = H_0 - \epsilon x=(D_x-By)^2+D_y^2,\,\,\,\,  Q=Q_0+V,$$ 
and in Lemma 1 we
will use the notation $H_1 = H$. For the simplicity we assume that
$\epsilon = B = 1$. The general case can covered by the same
argument.
\begin{lem}
Assume that $V, \partial_x V \in C^0(\R^2; \R) \cap L^\infty(\R^2;\R)$ and let $\psi\in C^\infty_0(\R^2)$. Then for $N\geq 2,\: j=0,1$ and for $\Im z\not=0$, the following operators are trace class: \\
i) $\psi (H_j\pm \ii)^{-N}$, \,\,\,$\partial_x \circ\psi (H_j\pm \ii)^{-N-1}$, \,\,\,\,$(H_j\pm \ii)\partial_x\circ\psi  (H_j\pm \ii)^{-N-2}$.\\
ii) $ (H_j\pm  \ii)^{-N}\psi$, \,\, $(H_j\pm \ii)^{-N-1}\psi\cdot\partial_x$.\\
iii)  $\psi \circ \partial_x(H_j\pm \ii)^{-N-1}$, \,\,\,\,$(H_j\pm \ii)\psi \circ\partial_x  (H_j\pm  i)^{-N-2}$.\\
iv) $(H_j\pm \ii) \partial_x (H_j\pm \ii)^ {-N-2}\psi$.\\
v) $(H_1+ \ii) \partial_x (H_1 + \ii)^ {-N-1}(H_1-z)^{-1}\psi$.\\
Moreover,
\begin{equation} \label{eq:2.1}
\Vert (H_1 +\ii) \partial_x (H_1 + \ii)^
{-N-1}(H_1-z)^{-1}\psi\Vert_{\tr}={\mathcal O}\Bigl(\frac{|z| + 1}{|\Im z|^2}\Bigr).
\end{equation}
\end{lem}
\begin{proof} 
We will prove the lemma only for $(H_1 + \ii)$, the case concerning $(H_1 - \ii)$ is similar. On the other hand, the statements for $(H_0 +\ii)$ follow from those for $(H_1 + \ii)$ when $V = 0$.

From the first resolvent equation, we obtain
$$ (H_1+z)^{-1}=(Q_0
+z)^{-1}-(Q_0 +z)^{-1}(x+V)( H_1+z)^{-1}$$
\begin{equation} \label{eq:2.2}
=(Q_0 +z)^{-1}+\sum_{j=1}^ {N+2}(-1)^j(Q_0 +z)^{-1}\Bigl((x+V)(Q_0
+z)^{-1}\Bigr)^{j}
\end{equation}
$$+(-1)^{N+3}\Bigl( (Q_0 +z)^{-1}(x+V)\Bigr)^{N+3}(H_1+z)^{-1}.$$
Taking $(N-1)$ derivatives with respect to $z$ in the above
identity and setting $z = \ii$, we see that $(H_1+\ii)^{-N}$ is a linear combination of terms
$$ {\mathcal K}_N:=(Q_0+\ii)^{-j_1}W(Q_0+\ii)^{-j_2}W...(Q_0+\ii)^{-j_r}W(H_1 +\ii)^{-p},$$
with $j_1+...+j_r\geq N,\:j_1 \geq 1,\: p\geq0$ and $W(x)=x +V(x)$. 

Recall that  if $P \in S^0(m)$ with $m\in L^1(\R^4),$  (resp. $m \in L^2(\R^4)$) then the
corresponding operator is trace class (resp. Hilbert-Schmidt).
By using this and  the fact that the symbol of $(Q_0 + \ii)^{-1}$ is
in $S^0(\langle \xi-y,\eta\rangle^{-2})$, we deduce that the operator  
$$K_{l,p,l',p'}^j:=\la x\ra^{-l}\la y\ra^{-p}(Q_0+\ii)^{-j}\la x\ra^{l'}\la y\ra^{p'}$$
 is trace class one for $l-l',p-p'>1, j\geq 2$ and 
Hilbert-Schmidt one for $l-l',p-p'>{1/2}, j\geq 1.$
Next, we write $\psi {\mathcal K}_N$ as follows
\begin{equation}\label{eq:2.3}
\psi {\mathcal K}_N=\psi \la x\ra^{3r} \la y\ra^{2r} K_{3r,2r,3r-2,2r-2}^{j_1} W\la x\ra^{-1}K_{3r-3,2r-2,3r-1,2r-4}^{j_2}W\la x\ra^{-1}
\end{equation}
$$...
W\la x\ra ^{-1} {\mathcal K}_{3,2,1,0}^{j_r}W\la x\ra^{-1}(H_1+\ii)^{-p}$$
Since $j_1+j_2+...+j_r\geq N\geq 2$, in the above decomposition, there are  at least two Hilbert-Schmidt operators or one of trace class. 
Combining this with the fact $\psi \la x\ra^{3r} \la y\ra^{2r}$, $W\la x\ra^{-1}$ and $(H_1+\ii)^{-p}$ are bounded from
$L^2(\R^2)$ into $L^2(\R^2)$, we conclude that $\psi {\mathcal K}_N$ is trace class operator. Thus $\psi(H_1+\ii)^{-N}$ is also a trace class operator. Repeating the same arguments, we obtain the proof for $\partial_x \circ\psi (H_j\pm \ii)^{-N-1}$.

As above to treat $(H_j\pm \ii)\partial_x\circ\psi  (H_j\pm \ii)^{-N-2}$, it suffices to show that
$(H_j\pm \ii)\partial_x\circ\psi {\mathcal K}_N$ is trace class. If we have $j_1\geq 2$ the proof is   completely similar to that of  $\psi(H_1+\ii)^{-N}$. In the case where $j_1=1$ since $(H_1+ \ii)\partial_x (Q_0+\ii)^{-1}$
is not bounded, we have to exploit the following representation
$$(H_1+ \ii)\partial_x\circ\psi {\mathcal K}_N=(H_1+\ii)(\partial_x\psi) {\mathcal K}_N$$
$$+ (H_1+\ii)\psi (Q_0+\ii)^{-1}\partial_x\circ W(Q_0+\ii)^{-j_2}W...(Q_0+\ii)^{-j_r}W(H_1+\ii)^{-p}.$$
Next use the fact that $\partial_xW \in L^\infty$ and repeat the argument of the proof above.

Recall that  $A$ is trace class if and only if the adjoint
operator $A^*$ is trace class. Consequently, (i) implies (ii).
Since $\psi\cdot\partial_x=\partial_x\cdot\psi-(\partial_x\psi)$,
the assertion (iii) follows from (i).

To deal with (iv), we apply the following obvious identity with $z = -\ii$,
\begin{equation}\label{eq:2.4}
\partial_x (H-z)^ {-1}=(H-z)^{-1}\partial_x+ (H-z)^{-1}(1+\pv)(H -z)^{-1},
\end{equation}
and obtain
\begin{equation}\label{eq:2.5}
(H_1 +\ii)\partial_x(H_1 +\ii)^{-N}\psi=(H_1 + \ii)^{-N}\partial_x\psi+\sum_{j=0}^{N-1} (H_1 +\ii)^{-j}(1+\pv)(H_1
+\ii)^{-N+j}\psi.
\end{equation}
Applying (i) and (ii) to each term on the right hand side of
(\ref{eq:2.5}), we get (iv).

Now we pass to the proof of (v). Applying (\ref{eq:2.4}), we obtain
$$(H_1+\ii) \partial_x (H_1 +\ii)^{-N-1}(H_1-z)^{-1}\psi=(H_1 + \ii) (H_1-z)^{-1}\partial_x (H_1+\ii)^ {-N-1}\psi$$
$$+ (H_1 + \ii)(H_1-z)^{-1}(1+\pv)(H_1-z)^{-1}(H_1+\ii)^{-N}\psi.$$
Combining the above equation with (i), (ii), (iv) and using the estimate
$$\Vert (H_1 + \ii)(H_1-z)^ {-1}\Vert={\mathcal O}\Bigl(\frac{|z| + 1}{\vert \Im z\vert}\Bigr),$$
we get (\ref{eq:2.1}).
\end{proof}
\begin{lem}Assume that $V(x,y)=\phi(x,y) W(x,y)$, where $\phi \in C^\infty_0(\R^2;\R)$ and $W, \partial_x W \in C^0(\R^2; \R) \cap L^\infty(\R^2;\R)$. Then for $N \geq 4$ the  operator
$$(H+\ii)\partial_x \Bigl[ (H+\ii)^{-N}-(H_0+\ii)^{-N}\Bigr],$$
is trace class.
\end{lem}
\begin{proof}
Taking ($N-1$) derivatives with respect to $z$ in the resolvent identity
$$(H+z)^{-1}-(H_0+z)^{-1}=-(H+z)^{-1}V(H_0+z)^{-1}$$
and setting $z =\ii$, we see that $(H+\ii)^{-N}-(H_0 + \ii)^{-N}$
is a linear combination of terms
$$(H+\ii)^{-j}V(H_0+\ii)^{-(N+1+j)}$$
with $1\leq j\leq N$. Composing the above terms by $(H
+\ii)\partial_x$ and applying Lemma 1, we complete the proof.
\end{proof}
\begin{lem}
Assume that $V$ satisfies the assumptions of Lemma $1$. Let $f\in C^\infty_0(\R)$ and $\psi\in C^\infty_0(\R^2)$. Then the operators $$\psi f(H_i),\,\,\, H_i \psi\partial_x f(H_i),\,\,\,\,
\psi\partial_x H_if(H_i)$$ are trace class and we have
$$ \tr \Bigl( H_i \psi\partial_x f(H_i)\Bigr)= \tr \Bigl(\psi\partial_x H_if(H_i)\Bigr).$$
\end{lem}
\begin{proof}
Set $g(x)=(x + \ii)^4f(x)$. Since $g(H_i)$ is bounded, it follows
from Lemma 1 that the operators
$$\psi(H_i + \ii)^{-4}g(H_i),\,\,\,H_i\psi\partial_x (H_i + \ii)^{-4} g(H_i),\,\,\,\,\psi\partial_x (H_i +\ii)^{-4} H_ig(H_i),$$
are trace class, and the cyclicity of the trace yields
$$\tr \Bigl( H_i \psi\partial_x f(H_i)\Bigr)= \tr \Bigl(H_i\psi\partial_x (H_i+\ii)^{-4} g(H_i)\Bigr)=
\tr \Bigl(H_ig(H_i)\psi\partial_x (H_i+\ii)^{-4} \Bigr)$$
$$= \tr \Bigl(\psi\partial_x (H_i+\ii)^{-4} g(H_i)H_i\Bigr)=
\tr \Bigl(  \psi\partial_x H_if(H_i)\Bigr).$$ Notice that in the
above equalities we have used the fact that the operators $g(H_i),
H_i$ and $(H_i + \ii)^{-4}$ commute.
\end{proof}
\begin{lem} Let $V$ be as in Lemma $2$. Then for every
$f\in C^\infty_0(\R)$  the operators
$$f(H)-f(H_0),\,\,\,\,\partial_x\Bigl(f(H)-f(H_0)\Bigr)\,\,\,\,\hbox { and } \,\,\,\,(H \pm \ii)\partial_x \Bigl(f(H)-f(H_0)\Bigr)$$ are trace class.
\end{lem}
\begin{proof}
Let $g(x)=(x + \ii)^4 f(x)$ be as above. We decompose
$$(H+\ii)\partial_x \Bigl(f(H)-f(H_0)\Bigr)=(H+\ii)\partial_x \Bigl((H+\ii)^{-4}-(H_0+ \ii)^{-4}\Bigr)g(H_0)+$$
$$
(H+\ii)\partial_x (H+\ii)^{-4}\Bigl(g(H)-g(H_0)\Bigr)= I + II.$$
According to Lemma 2, the operator $I$ is trace class. To treat $II$,
we use the Helffer-Sj\"ostrand formula
$$(II)=-{1\over \pi} \int \overline\partial \tilde g(z)(H+\ii)\partial_x (H+\ii)^{-4}\Bigl( (z-H)^{-1}-(z-H_0)^{-1}\Bigr) L(dz)$$
$$=-{1\over \pi} \int \overline\partial \tilde g(z)(H + \ii)\partial_x (H+\ii)^{-4}(z-H)^{-1}V(z-H_0)^{-1} L(dz),$$
where $\tilde{g}(z) \in C_0^{\infty}(\C)$ is an almost analytic continuation of $g$
such that $\overline\partial\tilde g(z)={\mathcal O}(\vert \Im
z\vert^\infty)$,
 while $L(dz)$ is the Lebesgue measure on $\C.$
Now applying Lemma 1, (v),  we see that the operator
$$(H+\ii)\partial_x (H+\ii)^{-4}(z-H)^{-1}V$$
 is trace class. Since $|z|$ is bounded on supp $\tilde{g}$, we can apply (\ref{eq:2.1}) to the right hand part of the above equation and combining this with $\overline\partial\tilde g(z)={\mathcal O}(\vert \Im z\vert^\infty)$, we deduce that $II$ is  trace class. Summing up, we conclude that $(H+\ii)\partial_x \Bigl(f(H)-f(H_0)\Bigr)$ is trace class. The same argument works for $(H - \ii)\partial_x \Bigl(f(H)-f(H_0)\Bigr)$. The proof concerning $f(H)-f(H_0)$ and $\partial_x\Bigl (f(H)-f(H_0)\Bigr)$ are similar and simpler.
\end{proof}
To establish Theorem 1, we also need the following abstract
result. For the reader convenience we present a proof.
\begin{prop}
Let $A$ be an operator of trace class on some Hilbert space $H$
and let $\{K_n\}$ be sequences of bounded linear operator which
converges strongly to $K\in {\mathcal L}(H)$. Then
$$ \lim_{n\rightarrow \infty} \Vert K_nA-KA\Vert_{\rm tr}=0.$$
\end{prop}

\begin{proof}
First assume that $A$ is a finite rank operator having the form
$A=\sum_{k=1}^ m<\cdot, \psi_k>\phi_k$, where $\psi_k, \phi_k \in
H$. Since
$$\Vert A\Vert_{\tr}\leq \sum_{k=1}^m \Vert \phi_k\Vert \Vert \psi_k\Vert,$$
we have
\begin{equation}\label{eq:2.6}
\Vert (K_n-K)A\Vert_{\rm tr}\leq \sum_{k=1}^m\Vert
(K_n-K)\phi_k\Vert \Vert \psi_k\Vert\rightarrow 0,\,\,\
n\rightarrow \infty.
\end{equation}
The general case can be covered  by an approximation. Since $K_n$
converges strongly, it follows from the Banach-Streinhaus theorem
that $\mu ={\rm sup}_n\Vert K_n \Vert<\infty$. Let $\eta$ be an
arbitrary positive constant and let $A_\eta$ be a finite rank
operator such that $\Vert A-A_\eta\Vert_{\tr} \leq {\eta\over
2\mu}$.  We have
$$\Vert (K_n-K)A\Vert_{\tr}\leq \Vert (K_n-K)(A-A_\eta)\Vert_{\tr}+\Vert (K_n-K)A_\eta\Vert_{\tr}\leq \eta+ \Vert (K_n-K)A_\eta\Vert_{\tr}.$$
Next we apply (\ref{eq:2.6}) for the finite rank operator $A_\eta$
and obtain
$$\lim_{n\rightarrow \infty} \Vert (K_n-K)A\Vert_{\rm tr}\leq \eta,$$
which implies Proposition 1, since $\eta$ is arbitrary.
 \end{proof}
{\bf Proof of Theorem 1.}  Assume first that $V=\phi W$ where $\phi\in
 C^\infty_0(\R^2;\R)$ and $W, \partial_x W \in C^0(\R^2; \R) \cap L^\infty(\R^2;\R)$. Choose a function $\chi \in
C_0^{\infty}(\R^2)$ such that $\chi = 1$ for $|(x,y)|\leq 1$. For
$R > 0$ set
$$\chi_R(x, y) = \chi\Bigl(\frac{x}{R}, \frac{y}{R}\Bigr),$$
and introduce
$$B_R:=[\chi_R\partial_x,H]f(H)-[\chi_R\partial_x,H_0]f(H_0).$$
Here $[A,B]=AB-BA$ denotes the commutator of $A$ and $B$.
According to Lemma 3, we have  $$\tr
\Bigl([\chi_R\partial_x,H]f(H)\Bigr)= \tr
\Bigl([\chi_R\partial_x,H_0]f(H_0)\Bigr)=0.$$ Thus
\begin{equation} \label{eq:2.7}
\tr (B_R)=0.
\end{equation}
On the other hand, a simple calculus shows that
\begin{equation}\label{eq:2.8}
B_R=\chi_R\Bigl([\partial_x,H]f(H)-[\partial_x,H_0]f(H_0)\Bigr)+[\chi_R,H_0]\partial_x\Bigl(f(H)-f(H_0)\Bigr):=B^1_R+B_R^2,
\end{equation}
where we have used that $[\chi_R,H]=[\chi_R,H_0]$.

Since $[\partial_x,H]=1+\partial_xV$ and $[\partial_x,H_0]=1$, it
follows from Lemma 3, Lemma 4 and Proposition 1 that
 \begin{equation}\label{eq:2.9}
 \lim_{R\rightarrow \infty}\tr(B_R^1)=\tr \Bigl(f(H)-f(H_0)\Bigr)+ \tr\Bigl(\pv f(H)\Bigr).
 \end{equation}
Next we claim that
\begin{equation}\label{eq:2.10}
\lim_{R\rightarrow \infty} B_R^2=0.
\end{equation}
Using that $[\chi_R,H_0]={2\over R} (D_x\chi_R) (D_x-y)- {2\over
R} (D_y\chi_R) D_y+{1\over R^2} (\Delta\chi_R)$, we decompose
$B_R^2$ as a sum of three terms $B_R^2 = I^1_R+I_R^2+I_R^3$, where
$$I_R^1=-{2\over R} (D_x\chi_R) (D_x-y)\partial_x\Bigl(f(H)-f(H_0)\Bigr),$$
$$I_R^2=- {2\over R} (D_y\chi_R) D_y\partial_x\Bigl(f(H)-f(H_0)\Bigr),$$
$$I_R^3={1\over R^2} (\Delta\chi_R)\partial_x\Bigl(f(H)-f(H_0)\Bigr).$$
To treat $I_R^1$, we set $Q = H - x$ and write
$$I_R^1=-{2\over R} (D_x\chi_R) (D_x-y)(Q_0-\ii)^{-1}(H-\ii)\partial_x\Bigl(f(H)-f(H_0)\Bigr)$$
$$+{2\over R} (D_x\chi_R) [(D_x-y)(Q-\ii)^{-1},x]\partial_x\Bigl(f(H)-f(H_0)\Bigr)$$
$$ +{2\over R} x(D_x\chi_R) (D_x-y)(Q-\ii)^{-1}\partial_x\Bigl(f(H)-f(H_0)\Bigr).$$
The operators $[(D_x-y)(Q-\ii)^{-1},x]$ and $(D_x-y)(Q-\ii)^{-1}$
are bounded, while $\partial_x\Bigl (f(H)-f(H_0)\Bigr)$ and
$(H-\ii)\partial_x\Bigl(f(H)-f(H_0)\Bigr)$  are trace class
operators (see Lemma 4). On the other hand, ${2\over R}
(D_x\chi_R)$, ${2\over R} x(D_x\chi_R)$ converges strongly to
zero. Indeed, since $\chi(x,y)=1$ for $|(x,y)| \leq 1$, we
get
$$\int \Bigl|{x\over R}(D_x\chi_R)u\Bigr |^ 2dxdy\leq {\rm sup}_{(x,y)\in \R^2}|xD_x\chi(x,y)|\int_{\{|(x,y)|\geq R\}} |u|^ 2dxdy\rightarrow 0,\,\, R\rightarrow \infty,$$
for all $u\in L^2(\R^2)$. Applying Proposition 1, we conclude that
\begin{equation}\label{eq:2.11}
\lim_{R\rightarrow \infty} I_R^1=0.
\end{equation}
To deal with $I_R^2,\: I_R^3$, notice that the operators $D_y (Q-
\ii)^{-1}$ and $[D_y(Q - \ii)^{-1}, x]$ are bounded and we repeat
the above argument. Thus we deduce
 \begin{equation}\label{eq:2.12}
\lim_{R\rightarrow \infty} I_R^{j}=0,\,\,\,j=2,3.
\end{equation}
Consequently, (\ref{eq:2.11}) and (\ref{eq:2.12}) imply
(\ref{eq:2.10}) and the claim is proved. Now, combining
(\ref{eq:2.7}), (\ref{eq:2.8}), (\ref{eq:2.9}) and
(\ref{eq:2.10}), we obtain Theorem 1 in the case where $V$ satisfies the assumption of Lemma 2 and $\epsilon=1$.

\begin{prop} Assume that $V \in L^\infty(\R^2;\R)$ satisfies $(1.1)$.
Then for $z \notin \R,\: z' \notin \R$ the operators $(z - H)^{-1}V(z' - H)^{-1},\: V(z - H)^{-1}(z' - H)^{-1},\:(H - z)^{-1} - (H_0 -z)^{-1}$ are trace class and
\begin{eqnarray} \label{eq:2.13}
\|(z - H)^{-1} V(z' - H)^{-1}\|_{\tr}\:\leq C_1 |\Im z|^{-1}|\Im
z'|^{-1},\\ \nonumber \| V(z - H)^{-1}(z' - H)^{-1}\|_{\tr}\:\leq
C_1 |\Im z|^{-1}|\Im z'|^{-1}.
\end{eqnarray}
Moreover, if $g \in C_0^{\infty}(\R),$ then the operator $V g(H)$ is trace class.
\end{prop}
\begin{proof}
Set $g_\delta(x,y)=\la x\ra^{-1-{\delta\over 2}}\la y\ra^{-{1+\delta\over 2}}$ and 
$ f_\delta(x,y)=\la x\ra^{-2-\delta}\la y\ra^{-1-\delta}$ ,
where $\delta$ is the constant in (1.1).
According to Lemma 8 in the Appendix,  $g_\delta(H_0+\ii)^{-1}$,
$(H_0+\ii)^{-1}g_\delta$ are Hilbert-Schmidt operators and $f_\delta (H_0+\ii)^{-2}$ is a trace one. 
Since $g_\delta^{-1}Vg_\delta^{-1}, \, Vf_\delta^{-1}\in L^\infty$, it follows that
$$(H_0+\ii)^{-1}V(H_0+\ii)^{-1}=(H_0+\ii)^{-1}g_\delta [g_\delta^{-1}Vg_\delta^{-1}]g_\delta(H_0+\ii)^{-1}$$
and $V(H_0+\ii)^{-2}$ are trace class operators. Next we write
$$(H + \ii)^{-1} - (H_0 + \ii)^{-1} = -(H_0 + \ii)^{-1}V(H_0 + \ii)^{-1} + (H + \ii)^{-1} V (H_0 + \ii)^{-1} V (H_0 + \ii)^{-1}$$
and conclude that $(H + \ii)^{-1} - (H_0 + \ii)^{-1}= -(H + \ii)^{-1}V(H_0 + \ii)^{-1}$ is trace class. 
Now consider the following equalities
 $$(\ii + H)^{-1}V(\ii + H)^{-1}=(\ii + H_0)^{-1}V(\ii + H_0)^{-1}+(\ii + H)^{-1}V(\ii + H_0)^{-1}V(\ii+H_0)^{-1}+$$
 $$(\ii + H_0)^{-1}V(\ii + H_0)^{-1}V(\ii +H)^{-1}+(\ii + H)^{-1}V(\ii + H_0)^{-1}V(\ii + H_0)^{-1}V(\ii + H)^{-1}$$
 and
$$V(H + \ii)^{-2} = V(H_0 + \ii)^{-2} - V(H_0 + \ii)^{-1} (H + \ii)^{-1} V (H_0 + \ii)^{-1} - V(H +\ii)^{-1} V (H_0 + \ii)^{-1} (H + \ii)^{-1}.$$
By using the trace class properties established above, we get   (\ref{eq:2.13}) for $z=z'=-i$. By applying the first resolvent equation 
$$(H-z)^{-1}=
(H+\ii)^{-1}+(\ii-z)(H+ \ii)^{-1}(H-z)^{-1},$$
 we obtain the general case.

 To examine $Vg(H)$, consider the function $h(x) = (x + \ii)^2g(x).$ Then $Vg(H) = V(H + \ii)^{-2}h(H)$ and since $V(H + \ii)^{-2}$ is trace class, we obtain the result.
\end{proof}
 For $R>0$ introduce
$$H_R:=H_0+\chi_R(x,y)V(x,y),$$
where $\chi_R(x,y)=\chi({x\over R}, {y\over R})$ with
$\chi\in C^\infty_0(\R^2)$ such that  $\chi=1$ in a neighborhood of $|(x, y)| \leq 1.$
\begin{rem} The result of Proposition $2$ concerning the trace class property of $(H-z)^{-1} - (H_0-z)^{-1},\: \Im z \neq 0,$ improves considerably Proposition $2$ in \cite{DP2}, where much more regular potentials have been examined. On the other hand, if the potential $V$ satisfies $(1.1)$ and $V, \partial_x V \in C^0(\R^2; \R)\cap L^{\infty} (\R^2; \R),$ then the statements of Proposition $2$ hold for the operators $(z - H_R)^{-1}V(z' - H)^{-1},\: z \notin \R, z' \notin \R$.
\end{rem}
The proof of Theorem 1 in the general case will be a simple
consequence of the following
\begin{lem} Let $V(x, y)$  be as in Theorem 1. Then for $f\in C^\infty_0(\R)$ we have
 \begin{equation}\label{eq:2.14}
 \lim_{R\rightarrow \infty} \tr \Bigl(f(H_R)-f(H)\Bigr)=0,
 \end{equation}
  \begin{equation}\label{eq:2.15}
\lim_{R\rightarrow \infty} \tr \Bigl (\partial_x(\chi_RV)
f(H_R)\Bigr)= \tr \Bigl(\partial_xV f(H)\Bigr).
\end{equation}
\end{lem}
\begin{proof}
Let $g(x)=(x+\ii) f(x)$ be as above. We decompose
$$
f(H_R)-f(H)=\Bigl((H_R+\ii)^{-1}-(H+\ii)^{-1}\Bigr)g(H)+
 (H_R +\ii)^{-1}\Bigl(g(H_R)-g(H)\Bigr)= J_R+ K_R.
$$
 From the first resolvent identity, we obtain
 $$J_R=(H_R-\ii)^{-1}(1-\chi_R)V(H+ \ii)^{-1}g(H)=(H_R-\ii)^{-1}(1-\chi_R)Vf(H).$$
 According to Proposition 2, the operator $Vf(H)$ is trace class and $ (H_R-\ii)^{-1}(1-\chi_R)$ converges strongly to zero. Then from Proposition 1 it follows that
  \begin{equation}\label{eq:2.16}
  \lim_{R\rightarrow \infty} \tr J_R=0.
  \end {equation}

 To treat $\tr K_R$, as in the proof of Lemma 4, we use the Helffer-Sj\"ostrand formula and write
$${\rm tr}\:K_R=-{1\over \pi} \int \overline\partial \tilde g(z){\rm tr}\Bigl((H_R+\ii)^{-1}\Bigl( (z-H_R)^{-1}-(z-H)^{-1}\Bigr)\Bigr) L(dz)$$
$$={1\over \pi} \int \overline\partial \tilde g(z){\rm tr}\Bigl((H_R+\ii)^{-1}(z-H_R)^{-1}(1-\chi_R)V(z-H)^{-1}\Bigr) L(dz).$$
By cyclicity of the traces we obtain
$${\rm tr}\Bigl((H_R+\ii)^{-1}(z-H_R)^{-1}(1-\chi_R)V(z-H)^{-1}\Bigr)={\rm tr}\Bigl((z-H_R)^{-1}(1-\chi_R)V(z-H)^{-1}(H_R+\ii)^{-1}\Bigr)$$
$$= {\rm tr}\Bigl((z-H_R)^{-1}(1-\chi_R)V(z-H)^{-1}(H+\ii)^{-1}\Bigr)$$
$$+{\rm tr}\Bigl((1 -\chi_R)V(H_R+\ii)^{-1}  (z-H_R)^{-1}(1-\chi_R)V(z-H)^{-1}(H+\ii)^{-1}\Bigr).$$
Now notice that for $z \notin \R$ the operators $(1
-\chi_R)V(H_R+\ii)^{-1}  (z-H_R)^{-1}(1-\chi_R)$ and
$(z-H_R)^{-1}(1-\chi_R)$ converge strongly to zero. On the other
hand, from Proposition 2 we deduce that the operator $V(z-
H)^{-1}(\ii + H)^{-1}$ is trace class. Thus for $z \notin \R$,
we conclude that the integrand converge to 0 as $R \to \infty.$ An
application of the Lebesgue convergence domination theorem
combined with the estimates (\ref{eq:2.13}) yield
\begin{equation}\label{eq:2.17}
  \lim_{R\rightarrow \infty} {\rm tr}\:K_R=0.
  \end {equation}
  Putting together (\ref{eq:2.16}) and (\ref{eq:2.17}), we obtain (\ref{eq:2.14}).

  Next, we pass to the proof of (\ref{eq:2.15}). A simple calculus shows that
  \begin{equation}\label{eq:2.18}
  \partial_x(\chi_RV) f(H_R)=\partial_x(\chi_RV) (f(H_R)-f(H))+{1\over R}(\partial_x\chi)_RVf(H)+(\chi_R\partial_xV f(H)).
   \end {equation}
Repeating the same arguments as in the proof of (\ref{eq:2.14}),
we show that
\begin{equation}\label{eq:2.19}
\lim_{R\rightarrow\infty} \tr\Bigl(\partial_x(\chi_RV)
(f(H_R)-f(H))\Bigr)=0.
 \end {equation}
On the other hand, since ${1\over R}(\partial_x\chi)_R$ (resp.
$\chi_R$) converges strongly to zero (resp.1), it follows from
Proposition 1 that
$$
\lim_{R\rightarrow \infty} \tr \Bigl({1\over
R}(\partial_x\chi)_RVf(H)\Bigr)=0,\,\,\,\lim_{R\rightarrow
\infty}\tr \Bigl(\chi_R\pv f(H)\Bigr)={\rm tr}\Bigl(\pv f(H)\Bigr),
 $$
 which together with (\ref{eq:2.18}) and (\ref{eq:2.19}) yield (\ref{eq:2.15}).
\end{proof}

{\bf End of the proof of Theorem 1.}
Applying Theorem 1 to $H_R$, we obtain : 
 $$\tr \Bigl[ f(H_R) - f(H)\Bigr]+ \tr \Bigl[ f(H) - f(H_0)\Bigr]= \tr \Bigl[ f(H_R) - f(H_0)\Bigr] = - \tr \Bigl(\d_x (\chi_RV )f(H)\Bigr),$$
 and an application of  Lemma 5 implies Theorem 1.

\section{Mourre estimate and embedded eigenvalues}


Consider the operator
$$Q = (D_x - By)^2 + D_y^2 + V(x, y),$$
 and set $\la x \ra = (1 + |x|^2)^{1/2},\: \la D_x \ra = (1 + D_x^2)^{1/2}.$

\begin{lem}
Assume that $V, \partial_x V \in C^0(\R^2; \R) \cap L^\infty({\R}^2;\R)$ and let $\Vert \1b_{\{\vert x\vert+\vert y \vert>R\}}(x,y) \partial_x V\Vert_{L^\infty }\rightarrow 0$ for $R\longrightarrow +\infty.$ Then for all $f \in
C^\infty_0 (\R)$, the operator $ f(H) \pv f(H)$ is compact.
\end{lem}

\begin{proof}
Let $\varphi(x,y)  \in C^\infty_0 (\R^2)$ be equal to one near zero. Set $\varphi_n(x,y)=\varphi({x\over n},{y\over n})$. 
According to  Lemma 3, the operator $f(H)\varphi_n \partial_xVf(H)$ is trace class. The
set of compact operators is closed with respect to the norm $\|
.\|_{\mathcal {\mathcal L}(L^2)}$ and the lemma follows from the obvious estimate

$$\begin{array}{l}
\| f(H) (1-\varphi_n) \partial_x V \, f(H) \|_{\mathcal L(L^2)}
\leq \| f^2(H) \|_{\mathcal L (L^2)} \| (1-\varphi_n) \partial_x V
\|_\infty.
\end{array}$$
\end{proof}

\begin{thm}
Let $[a, b] \subset \R.$  Under the assumptions of Lemma $6$, there
exists a compact operator $K$ such that
\begin{equation}\label{eq:3.2}
\1b_{[a, b]}(H) [\partial_x, H] \, \1b_{[a,b]}(H) \geq
\epsilon \1b_{[a,b]}(H) + \1b_{[a,b]}(H)K \1b_{[a,b]}(H).
\end{equation}
\end{thm}

\begin{proof}
Since the operator $\partial_x$ commutes with $(D_x-By)$ and
$D_y^2$, we have $[\partial_x,H]=\epsilon+\partial_xV$. Consequently,
\begin{eqnarray}\label{3.3}
\1b_{[a,b]}(H)[\partial_x,H]\1b_{a,b]}(H)=\epsilon\1b_{[a,b]}(H) +\1b_{[a,
b]}(H)\partial_{x}V\1b_{[a,b]}(H)\\ \nonumber
 =\epsilon \1b_{[a,b]}(H) + \1b_{[a,b]}f(H)\partial_x V f(H)\1b_{[a,b]}(H),
\end{eqnarray}
where $f \in C_0^{\infty}(\R)$ is a cut-off function such that $f= 1$ on $[a,b]$. Thus, Theorem 3 follows from Lemma 6.
\end{proof}
The use of commutators with the operator $\partial_x$ is well known for the analysis of the operator without magnetic field ($B = 0$) (see the pioneering work \cite{BCD} and \cite{AMG} for a more complete list of references). On the other hand, to treat crossed magnetic and electric fields we need Lemma 1 and Lemma 3. 

\begin{cor} In addition to the  assumptions of Theorem $3$ assume that $\partial_x^2 V\in C^0(\R^2) \cap L^\infty(\R^2)$. Then the point spectrum of $H$ in $[a,b]$ is finite and with finite multiplicity. Moreover, the singular continuous spectrum of $H$ is empty.
\end{cor}

\begin{proof}
Set $A=D_x$ and let $\alpha\in\R$. The explicit formula
$$e^{\ii\alpha A}(H+\ii)^{-1}=(e^{\ii\alpha A} He^ {-\ii\alpha A}+\ii)^{-1}e^{\ii\alpha A}=(H+\epsilon\alpha+V(x+\alpha,y)-V(x, y)+\ii)^{-1}e^{\ii\alpha A}$$
shows that  $e^{\ii\alpha A}$ leaves $D(H)$ invariant. On the other hand, since
$$\Vert He^{\ii\alpha A} (H+\ii)^{-1} \psi\Vert=\Vert e^{-\ii\alpha A} He^{\ii\alpha A} (H+\ii)^{-1} \psi\Vert$$
$$= \big\Vert \Bigl(H -\epsilon \alpha +V(x -\alpha, y)-V(x,y)\Bigr)(H+\ii)^{-1} \psi\big\Vert,$$
we deduce that  for each $\varphi\in D(H)$
$${\rm sup}_{\vert \alpha\vert <1} \Vert H e^{\ii\alpha A} \varphi\Vert<\infty.$$
Combining this with the fact $\ii[A,H]=\epsilon+\partial_x V$, $[A,[A,H]]=-\partial_x^2V$ and using (3.1), we conclude that the self-adjoint operator $A$ is a conjugate operator for $H$ at every $E \in \R$ in the sense of \cite{M}. Consequently, Corollary 1 follows from the main result in \cite{M} (see also \cite{AMG}, \cite{G}).

\end{proof}

\begin{rem}

For any sign-definite and bounded potential $V(x, y)$ such that $|V(x, y)| \to 0$ as $|x| + |y| \to \infty$ sufficiently fast in \cite{RW} and \cite{MR} it was established that  for $\epsilon=0$ the potential $V$ creates   an infinite number of
eigenvalues of $Q$  which accumulate to Landau levels. 
The above corollary  shows that only a finite number of these eigenvalues 
 may survive in the presence of a non vanishing constant electric field. In general, the problem of absence of embedded eigenvalues  when $\epsilon \neq 0$ remains open and this is an interesting conjecture.

\end{rem}

For a fixed value of $\epsilon\not=0$,  the following result shows that there are potentials for which  $H$ has absolutely continuous spectrum without embedded eigenvalues.

\begin{cor}
Fix $\epsilon>0$.  Assume that $\partial_x^\alpha V \in C^0(\R^2; \R) \cap L^ \infty(\R^ 2;\R),\:\alpha=0,1,2$
and
\begin{equation} \label{eq:3.3}
\epsilon+\partial_x V(x,y)>c>0,
\end{equation}
uniformly on   $(x,y)\in \R^ 2$. Then $H$ has no eigenvalues. Moreover, for $s>{1/2}$,  the following estimates holds uniformly on $\lambda$ in a compact interval
\begin{equation} \label{eq:3.4}
\| \la D_x \ra ^{-s} (H - \lambda \pm \ii 0)^{-1} \la D_x \ra
^{-s} \| = {\mathcal O}_{\epsilon}(1).
\end{equation}

\end{cor}

\begin{proof}
Let $[a,b]$ be a compact interval in $\R$. From (\ref{eq:3.2}) and (\ref{eq:3.3}), we have 
\begin{equation} \label{eq:3.5}
\1b_{[a,b]}(H)[\partial_x,H]\1b_{a,b]}(H)\geq c \1b_{[a,b]}(H).
\end{equation}
According to the proof of Corollary 1, $A=D_x$ is a conjugate operator in the sense of \cite{M}. Combining this with 
(\ref{eq:3.5})
we deduce from \cite{M} that $H$ has no eigenvalue in $\R$. Applying once more Mourre theorem (see \cite{M}, \cite{AMG}, \cite{G}), we obtain the estimate (\ref{eq:3.4}).
\end{proof}
\section{Limiting absorption principle}
\renewcommand{\theequation}{\arabic{section}.\arabic{equation}}
\setcounter{equation}{0}

\def\la{\langle}
\def\ra{\rangle}
\def\lx{\la x \ra}

 In this section we treat the case when $\epsilon$ is small enough. Notice that when  $\epsilon$ tends to zero in general the assumption $\epsilon + \partial_x V > c > 0$ is not satisfied and we cannot apply Corollary 2. Our goal is to study the behavior of the resolvent $(H-\lambda \pm \ii \delta)^{-1}$ as $\delta \to 0$ for $\lambda \notin \sigma(Q)$. For such $\lambda$ we could have eigenvalues of $H$ and a direct application of Mourre argument is not possible. We will obtain the result assuming that $\epsilon $ is small and for this purpose we need the following

\begin{lem} Assume  that $V\in L^ \infty(\R^2;\R)$ and let $\lambda \notin \sigma(Q)$. Let $\chi \in C_0^{\infty}(\R; \R)$ be equal to $1$ near $\lambda$ and
let ${\rm supp}\: \chi \cap \sigma(Q) = \emptyset.$ Then
\begin{equation} \label{eq:4.1}
\Vert \chi(H) \lx^{-2}\Vert \leq C \epsilon^2.
\end{equation}
\end{lem}
\begin{proof} Since ${\rm \supp}\: \chi \cap \sigma(Q) = \emptyset,$ the operators $(z - Q)^{-1}$ and
$(z - Q)^{-1}  x (z - Q)^{-1}$ are analytic operator valued functions for $z$ in a complex
neighborhood of supp $\chi$. Let $\tilde{\chi}(z) \in
C_0^{\infty}(\C)$ be an almost analytic continuation of $\chi(x)$
such that
$$\bar{\partial} \tilde{\chi}(z) = {\mathcal O}(|\Im z|^{\infty})$$
and ${\rm supp}\: \tilde{\chi}(z) \cap \sigma(Q) = \emptyset.$ We
have the representation
$$\chi(H) = -\frac{1}{\pi} \int \bar{\partial} \tilde{\chi}(z) (z - H)^{-1} L(dz),$$
where $L(dz)$ is the Lebesgue measure in $\C$. By using the
resolvent identity, we get
$$(z-H)^{-1}=(z-Q)^{-1}+\epsilon (z - Q)^{-1}  x (z - Q)^{-1}+\epsilon^2
(z - H)^{-1} x (z - Q)^{-1}  x (z - Q)^{-1},$$ and we obtain
$$\chi(H) = \chi(Q)-\frac{\epsilon}{\pi} \int \bar{\partial} \tilde{\chi}(z) (z - Q)^{-1}  x (z - Q)^{-1}L(dz)$$
$$
-\frac{\epsilon^2}{\pi} \int \bar{\partial} \tilde{\chi}(z) (z -
H)^{-1} x (z - Q)^{-1}  x (z - Q)^{-1}L(dz).$$ Since  ${\rm
supp}\: \tilde{\chi}(z) \cap \sigma(Q) = \emptyset,$ the first two
terms on the right hand side vanish. Consequently,
\begin{equation} \label{eq:4.2}
 \chi(H)=
-\frac{\epsilon^2}{\pi} \int \bar{\partial} \tilde{\chi}(z) (z -
H)^{-1} x (z - Q)^{-1}  x (z - Q)^{-1}L(dz).
\end{equation}

Next we observe that
$$x(z - Q)^{-1} = (z - Q)^{-1}x  + (z - Q)^{-1}[x, Q](z - Q)^{-1} = (z - Q)^{-1}x + L_1.$$
We have $[x, Q] = 2(D_x - By)$. Thus it is easy to see that for $ z
\notin \sigma(Q)$, $L_1 = (z - Q)^{-1}[x, Q](z - Q)^{-1}$ is a
bounded operator since $(D_x - By)(\ii - Q)^{-1}$ is bounded and
$(z - Q)^{-1} = (\ii - Q)^{-1} + (\ii - Q)^{-1}(\ii - z)(z -
Q)^{-1}$. We write
$$x(z - Q)^{-1} x (z - Q)^{-1} = (z - Q)^{-1} x (z - Q)^{-1} x $$
$$+ (z - Q)^{-1} x L_1 + L_1(z - Q)^{-1}x + L_1^2 = \sum_{j = 1}^4 I_j.$$
The operators $I_4 = L_1^2$ and $I_3 = L_1 (z -Q)^{-1} x \la x \ra
^{-2}$ are bounded. To see that $I_1 \la x \ra^{-2}$ is bounded, note that
$$I_1 \la x \ra ^{-2} = (z - Q)^{-2} x^2 \la x \ra^{-2} + (z - Q)^{-1}L_1 x \la x \ra ^{-2}.$$
Finally,
$$I_2\la x \ra ^{-2} = (z - Q)^{-2}x [x, Q](z - Q)^{-1} \la x \ra ^{-2} + (z - Q)^{-1} L_1 [x, Q](z - Q)^{-1}\la x \ra ^{-2}$$
and since the second term on the right hand side is bounded, it
remains to examine the operator
 $$x [ x, Q] (z - Q)^{-1} \la x \ra ^{-2} = [x, Q] x (z - Q)^{-1} \la x \ra ^{-2} + 2\ii (z - Q)^{-1} \la x \ra ^{-2}.$$
 Applying the above argument, we see that the last operator is bounded. Consequently, the operator under integration in (\ref{eq:4.2}) is bounded by ${\mathcal O}(|\Im z|^{-1})$ and this proves the statement.
\end{proof}
\begin{prop} Assume  that $\partial_x^\alpha V \in C^0(\R^2; \R) \cap L^ \infty(\R^2;\R)$ for $\alpha = 0, 1, 2$ and let $\la x\ra^2\partial_xV\in L^\infty(\R^2)$. Let $[a, b]$ be a compact interval such that $[a, b] \cap \sigma(Q) = \emptyset.$ Then for $s > 1/2$ and
sufficiently small $\epsilon_0 > 0$ we have the following estimate
uniformly with respect to $\lambda \in [a, b]$ and $\epsilon \in ]0, \epsilon_0]$
\begin{equation} \label{eq:4.3}
\| \la D_x \ra ^{-s} (H - \lambda \pm \ii 0)^{-1} \la D_x \ra
^{-s} \| \leq C\epsilon^{-1}.
\end{equation}
Moreover, $H$ has no embedded eigenvalues and singular continuous spectrum in $[a, b].$
\end{prop}
\def\1b{{\mathbb I}}
\begin{proof} Let $[a - \delta, b + \delta] \cap \sigma (Q) = \emptyset$ for $0 < \delta \ll 1.$ Choose a function
$\chi(t) \in C_0^{\infty}(\R ; \R)$ such that supp $\chi \subset
[a-\delta, b +\delta]$ and $\chi(t) = 1$ for $a_1 = a -\delta/2
\leq t \leq b + \delta/2 = b_1.$ Then
$$\1b_{[a_1, b_1]}(H) [\partial_x, H] \1b_{[a_1, b_1]}(H) = \epsilon \1b_{[a_1, b_1]}(H) + \1b_{[a_1, b_1]}(H) \partial_x V \1b_{[a_1, b_1]}(H)$$
$$ =  \epsilon \1b_{[a_1, b_1]}(H) + \1b_{[a_1,b_1]}(H)\Bigl(\chi(H) \la x \ra ^{-2}\Bigr) \Bigl(\la x \ra ^2 \partial_x V \Bigr)\1b_{[a_1, b_1]}(H)$$
Our assumption implies that the multiplication operator $\la x \ra
^2 \partial_xV\in L^\infty$ , while Lemma 7 says that
$$\|\chi(H) \la x \ra ^{-2}\| \leq C \epsilon^2.$$
Thus
$$\1b_{[a_1,b_1]}(H)\Bigl(\chi(H) \la x \ra ^{-2}\Bigr) \Bigl(\la x \ra ^2 \partial_x V \Bigr)\1b_{[a_1, b_1]}(H) \leq C_1 \epsilon^2 \1b_{[a_1,b_1]}(H)$$
and with a constant $c_0 > 0$ we deduce
$$\1b_{[a_1, b_1]}(H) [\partial_x, H] \1b_{[a_1, b_1]}(H) \geq c_0 \epsilon \1b_{[a_1,b_1]}(H).$$
Then it is well known (see for instance \cite{M}, \cite{AMG}, \cite{G}) that
for $\lambda \in [a, b]$ we get (\ref{eq:4.3}) and $H$ has no eigenvalues and singular continuous spectrum in $[a, b].$
\end{proof}
\begin{rem} As we mentioned in Remark $2$ for sign-definite rapidly decreasing potentials the spectrum of the operator $Q$ is formed by infinite  number eigenvalues having as points of accumulation the Landau levels $\mu_n = (2n + 1)B,\: n \in \N.$  For such potentials Proposition $3$ shows that the embedded eigenvalues of $H$ could appear only
in small neighborhoods of the eigenvalues of $Q$. Since in every interval we may have only a finite number of eigenvalues  of $H$, it is clear that for some  eigenvalues $\nu$ of $Q$ there are no eigenvalues
of $H$ in their neighborhoods. Moreover, it was proved in \cite{KP} that for potentials $V \in C_0^{\infty}(\R^2)$ we have
$\sigma(Q) \cap ]\mu_n - B , \mu_n + B[ \subset (\mu_n - Cn^{-1/2}, \mu_n + Cn^{-1/2}),\: n \geq N$ with $C > 0$ and $N$ depending only on $\sup|V|$ and the diameter of the support of $V$. Thus for $M$ large the embedded eigenvalues $\lambda \geq M $ of $H$ are sufficiently close to Landau levels $\Lambda_n.$

\end{rem}
\section{Estimates for the derivative of the spectral shift function}
\def\ds{\la D_x \ra^s}
\def\dss{\la D_x \ra^{2s}}
\def\pv{\partial_x V}
\def\qz{(Q - z)^{-1}}
\def\xx{\la x \ra}
First we notice that the assumption (\ref{eq:1.4}) makes possible to define the spectral shift function $\xi(\lambda,
\epsilon)$ related to operators $H_0(\epsilon) = H_0(B, \epsilon)$
and $H(\epsilon) = H_0(B, \epsilon) + V(x, y)$ by the equality
$$ \la \xi', f \ra = \tr \Bigl( f(H(\epsilon)) - f(H_0(\epsilon))\Bigr),\: f \in C_0^{\infty}(\R).$$
Here and below we omit the dependence of $B$ in the notations. Our
purpose in this section is to establish Theorem 2. For the proof we need the following
\begin{prop}
Under the assumptions of Theorem $2$, for $\lambda_0 \notin \sigma(Q)$ and $1/2 < s < \min{(1/2+\delta/4,1)}$ the
operator
$$\ds \pv \Bigl[\qz x\Bigr]^n \ds$$
is trace class for $z$ in a small complex neighborhood $\Xi
\subset \C$ of $\lambda_0.$
\end{prop}
\begin{proof} Before starting the proof, notice that it is easy to establish the statement for $z \ll 0$ since in this case the operator $\qz$ is a pseudodiferential one and we can apply the calculus of pseudodifferential operators and the criteria which guarantees that a pseudodifferential operator is trace class (see for instance, \cite{DS}, Theorem 9.4). For $z \in \R^+ \setminus \sigma(Q)$ this is not the case and $\qz$ is a bounded operator but not a pseudodifferential one. We may replace $\qz$ by the pseudodifferential operator $(Q - \ii)^{-1}$ modulo bounded operators but therefore  it is difficult to examine the product involving many bounded operators and factors $x^k$. To overcome this difficulty, we are going to apply a convenient decomposition by product of operators having in mind that the operator on the left of a such product must be trace class one. \\

 First we treat the case $n = 2$, the general case will be covered by a recurrence. We start with the analysis of the operator
\begin{equation} \label {eq:5.1}
\dss \pv [\qz x]^2.
\end{equation}
 Our goal is to show that (\ref{eq:5.1}) is a trace class operator. Write
$$\dss \pv \xx^2 \xx^{-2} \qz x \qz x = \dss (\pv )\xx^2 \qz \xx^{-2} x \qz x$$
$$+ \dss \pv \xx^2 \qz [Q, \xx^{-2}] \qz x \qz x$$
$$= \dss \pv \xx^2 (Q -z)^{-2} \Bigl[\xx^{-2} x^2 +  [Q, \xx^{-2} x] \qz x \Bigr]$$
$$+ \dss \pv \xx^2 \qz [Q, \xx^{-2}] \qz x \qz x= T_1 + T_2.$$
To deal with $T_1$, we use the representation
$$T_1 =  \dss \pv \xx^2 (Q -z)^{-2} W_1$$
 and we will show that the operator
$$W_1 = \xx^{-2} x^2 + [Q, \xx^{-2} x] \qz x $$
$$= \xx^{-2}x^2 - \ii \Bigr[(D_x - By) \frac{1 -x^2}{(1 + x^2)^2} + \frac{1 - x^2}{(1 + x^2)^2} (D_x - By)\Bigr](Q - z)^{-1} x$$
is bounded. Consider the operator
$$(D_x - By) \frac{(1- x^2)}{(1 + x^2)^2} \qz x = (D_x - By) \frac{(1- x^2)x}{(1 + x^2)^2}  (Q - \ii)^{-1}\Bigl[1 + (z - \ii)\qz\Bigr]$$
$$+(D_x - By) \frac{1 - x^2}{(1 + x^2)^2}\qz [Q, x] \qz.$$
The pseudodifferential operator
$$(D_x - By) \frac{(1- x^2)x}{(1 + x^2)^2} (Q - \ii)^{-1}$$
is bounded and the product of this operator with $\Bigl[1 +
(\ii-z)\qz\Bigr]$ is bounded, too. As in the proof of Lemma 7, we
see that $[Q, x] \qz$ is bounded and with the same argument we
treat the other terms. Thus we conclude that $W_1$ is a bounded
operator. Next we write
$$T_2 = \dss \pv \xx^2 (Q - z)^{-2} W_2,$$
where
 $$W_2 = [ Q, \xx^{-2}] x \qz x + \Bigl[Q, [Q, \xx^{-2}]\Bigr] \qz x \qz x = W_{21} + W_{22}.$$
We have
$$W_{21}=2 \ii \Bigl[ (D_x - By) \frac{x^2}{(1 + x^2)^2} \qz x + \frac{x}{(1 + x^2)^2} (D_x - By) x \qz x\Bigr]$$
and as above we deduce that $W_{21}$ is a bounded operator.
For the analysis of $W_{22}$, we write
$$W_{22} = \Bigl\{ \frac{1 - 3x^2}{(1+ x^2)^3} 4(D_x - By)^2 + R_1(x)(D_x -By) + R_2(x) + \frac{x}{(1 + x^2)^2}(4\partial_x V + 8B D_y)\Bigr\}$$
$$\times  \qz x \qz x.$$
A simple calculus
gives
$$\qz x \qz x = \qz x^2 \qz  + \qz x  M_1$$
$$= x^2 (Q-z)^{-2} + 4\qz x(D_x - By) (Q-z)^{-2} + x \qz M_1 + \qz M_2$$
$$= x^2 (Q-z)^{-2} + 4x \qz M_3 + \qz M_4$$
$$= x^2 (Q-\ii)^{-2}M_5 + 4x (Q-\ii)^{-1} M_6 + (Q-\ii)^{-1} M_7,$$
where  $M_k\:, k =1,2....$, denote bounded operators. The
pseudodifferential calculus implies  that the product of the term in the brackets $\{...\}$ with $x^j(Q- \ii)^{-j},\: j = 1,2$ is a bounded operator. Combining
this with the above equality, we conclude that $W_{22}$ is bounded.\\

Now it remains to see that the operator
$${\mathcal T} = \dss \pv \xx^2 (Q - z)^{-2}$$
is trace class. For this purpose we replace $(Q - z)^2$ by
$$(Q - \ii)^{-2}\Bigl[ I + (z- \ii) \qz\Bigr]^2$$
and consider the pseudodifferential operator
\begin{equation} \label{eq:5.2}
\dss \pv \xx^2 (Q - \ii)^{-2}
\end{equation}
with principal symbol
$$g_s(x, y, \xi, \eta) = \frac{\xi^{2s} (\pv)(x,y) (1 + x^2)}{\Bigl((\xi - By)^2 + \eta^2 + V(x,y) - \ii\Bigl)^2}.$$
We use the estimate $\la \xi \ra^{2s} \leq C\la \xi - By \ra^{2s} \la y \ra^{2s}$
and we apply Theorem 9.4 in \cite{DS} to deduce that (\ref{eq:5.2}) is a trace class operator. In fact we have 
$$\sum_{|\alpha| \leq 5} \|\partial_{x,y,\xi, \eta}^{\alpha} g_s\|_{L^1(\R^4)} < \infty$$
since $2s < 2$ guarantees that the integral with respect to $\xi$ is convergent, while
$2s < 1 + \delta/2$ and the estimate (\ref{eq:1.4}) imply that integral  with respect to $y$ is convergent.
Consequently, ${\mathcal T}$ is a trace class operator and this completes the
analysis of (\ref{eq:5.1}). Notice  also that the same argument
implies that the operator
$$\ds \pv \Bigl[\qz x\Bigr]^2$$
is trace class.\\
To prove that the operator $\ds \pv \Bigl[\qz x\Bigr]^2 \ds$ is
trace class, we commute the operator $\ds$ with $\qz x$ and $\pv$ in
order to reduce the proof to that of (\ref{eq:5.1}). The
commutators $[x, \ds]$ and $[V,\ds]x$  are bounded since $s < 1$. Next
$$[\qz , \ds]x = \qz [V,\ds]\qz x$$
$$ =  \qz [V,\ds]  \ra  \Bigl(x \qz + \qz M_1\Bigr) = \qz M_2$$
and we obtain operators which can be handled by the above argument. Thus the assertion is proved for $n = 2$.\\
\def\st{\sim_t}
\def\qzz{(Q - z)^{-2}}

Passing to the general case $n > 2$, assume that  the assertion
holds for $n = 2,...,k-1,$
 and suppose that $V$ satisfy the estimate (\ref{eq:1.4}) with $n = k$. The idea is to replace the
operator
$$\ds \pv [\qz x]^{k}\ds$$
by the trace class operator $\ds (\pv) x^k (Q-z)^{-2} \ds$ plus a
sum of several operators which are trace class according to the
recurrence assumption. Notice that if $M_j$ is bounded operator
obtained as a product of $(D_x - By)$ and $(Q-z)^{-j},\: j \geq
1$, the operator $\la D_x \ra^{-s} M_j \ds$ becomes a bounded
operators and this makes possible to exploit the representation
$$\ds \pv \qz x....M_j \ds = \Bigl[\ds \pv \qz x....\ds \Bigr] \Bigl(\la D_x \ra^{-s} M_j \ds\Bigr)$$
Thus we reduce the analysis to the trace class property of $\ds \pv \qz x....\ds$. For simplicity of the notations we will write $A \st B$ if the difference $A- B$ is a trace class operator.\\

We start with the observation that
$$\ds \pv [\qz x]^{k} \ds \st \ds \pv [\qz x]^{k-2} \qz x^2 \qz \ds.$$
We can establish this by a recurrence. For $k-1$ we apply the
equality
$$\ds \pv [\qz x]^{k-1} \ds =\ds \pv [\qz x]^{k-3}\qz x^2 \qz \ds$$
$$\ds \pv [\qz x]^{k-2} \qz [Q, x]\qz\ds $$
$$\st \ds \pv [\qz x]^{k-3}\qz x^2 \qz \ds.$$
Commuting $\qz$ and $x^2$, we obtain the result for $k-1$ and in the same way we continue for $p \leq k-1.$\\

Next we commute $\qz$ and $x^2$ and get
$$\ds \pv [\qz x]^{k-2} \qz x^2 \qz \ds$$
$$ \st \ds \pv [\qz x]^{k-3} \qz x^3 \qzz \ds.$$
Indeed,  $[Q, x^2] = 4(D_x - By) x = -4\ii x(D_x - By) -2$ yields
$$\qz x^2 \qz = x^2 \qzz -4\ii\qz x(D_x - By) \qz -2\qzz$$
and for the term
$$\ds \pv [\qz x]^{k-1} (D_x-By) (Q-z)^{-1} \ds$$
we use the recurrence assumption and the fact that $M_2 = (D_x -
By) \qz$ is a bounded operator. In the same way for $1 \leq j \leq
k-1$ we show that
$$\ds \pv [\qz x]^{k-j} \qz x^j \qzz \ds$$
$$\st \ds \pv [\qz x]^{k-j-1}\qz x^{j+1} \qzz \ds,$$
taking into account the equality
$$[Q, x^j] = 2j(D_x - By)x^{j-1} = 2jx^{j-1}(D_x - By) - 2\ii j(j-1)x^{j-1}$$
and the recurrence assumption. Finally, we prove that
$$\ds \pv [\qz x]^k \ds \st \ds (\pv) x^k \qzz \ds$$
and, as in the proof in the case $n = 2$, we conclude that the
operator on the right hand side is trace class one.

\end{proof}

After this preparation we pass to the proof of Theorem 2.\\

{\bf Proof of Theorem 2}. Let $\Xi \subset \R$ be a small
neighborhood of $\lambda_0$ such that $\Xi \cap \sigma(Q) =
\emptyset.$ For the simplicity of the notations we will write
$H(\epsilon),\: \xi(\lambda, \epsilon)$ instead of $H(B,
\epsilon), \: \xi(\lambda; B, \epsilon).$ Given $f \in
C_0^{\infty}(\Xi),$ introduce an almost analytic continuation
$\tilde{f} \in C_0^{\infty}(\C)$ of $f$ so that
$\bar{\partial}\tilde {f}(z) = {\mathcal O}(|\Im z|^{\infty})$ and
${\rm supp}\: \tilde{f}(z) \cap \sigma(Q) = \emptyset$. Since $(z
- Q)^{-1}$ is analytic over the support of $\tilde{f}(z)$,
applying the resolvent equality, we get
\begin{eqnarray} \label{eq:5.7}
\partial_x V f(H(\epsilon)) = -\frac{1}{\pi}\int \bar{\partial}\tilde {f}(z) \partial_x V (z -H(\epsilon))^{-1} L(dz)\\ \nonumber
= (- 1)^{n+1}\frac{\epsilon^n}{\pi} \int \bar{\partial}\tilde
{f}(z) \partial_x V [(z - Q)^{-1} x]^n (z - H(\epsilon))^{-1}
L(dz).
\end{eqnarray}
Taking into account Proposition 4 and the cyclicity of the trace,
we get
$$\tr\int \bar{\partial}\tilde {f}(z) \la D_x \ra ^{-s} \Bigl[\la D_x \ra ^s  \partial_x V [(z - Q)^{-1} x]^n \la D_x \ra^s \Bigr]\la D_x \ra^{-s}(z - H(\epsilon))^{-1} L(dz)$$
$$ = \tr \int \bar{\partial}\tilde {f}(z)\Bigl[\la D_x \ra ^s  \partial_x V [(z - Q)^{-1} x]^n \la D_x \ra ^s \Bigr]\la D_x \ra ^{-s} (z - H(\epsilon))^{-1} \la D_x \ra ^{-s}L(dz).$$

Set $W(z) = \la D_x \ra ^s  \partial_x V [(z - Q)^{-1} x]^n \la
D_x \ra ^s $ and note that for $z \in {\rm supp}\: \tilde{f}$ this
operator is trace class and $W(z)$ is analytic. We write
$$ -\frac{1}{\pi}\int \bar{\partial}\tilde {f}(z) \tr\Bigl(\partial_x V [(z - Q)^{-1} x]^n (z - H(\epsilon))^{-1}\Bigr) L(dz)$$
$$= \frac{1}{\pi}\lim_{\eta \searrow 0}\Bigl[\int_{\Im z > 0} \bar{\partial} \tilde{f}(z +\ii \eta)\tr \Bigl[\Bigl(W(z + i \eta) \la D_x \ra^{-s}(H(\epsilon) - (z + \ii \eta))^{-1}\la D_x \ra ^{-s}\Bigr)\Bigr] L(dz)$$
$$ + \int_{\Im z < 0} \bar{\partial} \tilde{f}(z - \ii\eta)\tr \Bigl(W(z - \ii\eta)\la D_x \ra^{-s}(H(\epsilon) - (z - \ii\eta) )^{-1}\la D_x \ra ^{-s} \Bigr)  L(dz)\Bigr].$$
Notice that the functions
$$\tr\Bigl(W(z \pm \ii\eta)\la D_x \ra^{-s}(H(\epsilon) - (z  \pm \ii\eta))^{-1}\la D_x \ra ^{-s}\Bigr)$$
are analytic in $\pm \Im z > 0$. Applying Green formula,  as in
Lemma 1 in \cite{DP1}, we deduce
$$\la \xi'(\lambda, \epsilon), f \ra = {\rm tr} \Bigl(f(H(\epsilon) - f(H_0)\Bigr) = -\frac{1}{\epsilon} {\rm tr} \Bigl (\pv f(H(\epsilon)\Bigr)$$
$$= \lim_{\eta \searrow 0} \frac{(-1)^{n}\epsilon^{n-1}}{2\pi \ii} \int f(\lambda)\tr \Bigl(W(\lambda) \Bigl[\la D_x \ra ^{-s}\Bigl((H(\epsilon) - (\lambda + \ii\eta))^{-1} - (H(\epsilon) - (\lambda - \ii\eta))^{-1}\Bigr) \la D_x \ra ^{-s} \Bigr]\Bigr)d\lambda,$$
where the integral is taken in the sense of distributions. On the
other hand, Proposition 4 combined with (\ref{eq:4.3}) show that
the right hand side of the above representation is finite and has
order ${\mathcal O}(\epsilon^{n-2}).$ Thus for $\forall f \in
C_0^{\infty}(\Xi)$ we obtain
$$\la \xi'(\lambda, \epsilon), f \ra = \int f(\lambda) T_{\epsilon}(\lambda) d\lambda$$
with $T_{\epsilon}(\lambda) = {\mathcal O}(\epsilon^{n -2})$ and
this completes the proof.
\section{Appendix}

The proof of the following Lemma is similar to the proof  of Proposition 2.1 in \cite{DP2} and for the reader convenience we give it.
\begin{lem}
Let $\delta>0$ and let $k_j(x,y)=\la x\ra^{-j(1+\delta)}\la y\ra^{-j({1\over 2}+\delta)}, j=1,2$. The operators
$G_2:=k_2(H_0+\ii)^{-2},\,\,G_2^*$,
(resp. $G_1:=k_1(H_0+\ii)^{-1},\,\,\, G_1^*$), are trace class (resp. Hilbert-Schmidt).
\end{lem}
\begin{proof}

Without loss of the generality we may assume that $B=\epsilon=1$.   Introduce the unitary operator $U : L^2(\R^2)\rightarrow  L^2(\R^2)$ by
$$(Uu)(x,y)={2\over \pi} \iint_{\R^2} e^{i\varphi(x,y,x',y')} \: u(x',y') \: dx'dy',$$
where
$\varphi(x,y,x',y')=xy-xy'-x'y+x'y'-\frac{1}{2} y'$.
A simple calculus shows that
\[ \widetilde H_0 = U^{-1} H_0 U = (D_y^2 + y^2) + x -\frac{1}{4} \,,\]
\[ \tilde k_j^\omega = U^{-1}k_j U = k_j^{\omega}\Bigl(x -  D_y - \frac{1}{2},  y +  D_x\Bigr) \,.\]

Since $U$ is unitary, it suffices to prove the lemma for $\tilde G_j:=UG_jU^{-1}=\tilde k_j^\omega(\widetilde H_0+\ii)^{-j}$.

   Let $\chi(t) \in C_0^{\infty} (\R; [0, 1])$ be a cut-off function such that $\chi(t) = 1$ for $|t| \leq 1$ and $\chi(t) = 0$ for $|t| \geq 2.$ Fix a number $k$, $\max \{1, \frac{2}{1 +2\delta}\} < k < 2,$ and introduce the symbol
\[ q(x, y, \eta) = \chi\Bigl( \frac{\la y, \eta \ra ^k}{|\eta^2 + y^2 + (x + i)|}\Bigr) \,, \]
where $\la y, \eta \ra = (1 + y^2 + \eta ^2)^{1/2}.$ It clear that $q(x, y, \eta) \in S^0(\R^4_{(x,\xi, y, \eta)})$
and we set $A = q^{\omega}(x, y, D_y).$ We decompose
\begin{equation}\label{eq:6.1}
 \tilde k_j^\omega(\widetilde H_0 + \ii)^{-j}  \, =A \tilde k_j^{\omega} (\widetilde H_0 +\ii)^{-j} + (I- A) \tilde{k}_j^{\omega}(\widetilde H_0 + \ii)^{-j}=L_j+M_j. \, \end{equation}
To treat $L_j$, notice that on the support of $q(x, y, \eta)$ we have 
\[ (\eta^2 + y^2 + x + \ii)^{-1} \in S^0(\R^4; \la y, \eta \ra ^{-k}) \,. 
\]
In fact, on the support of $q$ we obtain
\[ \la y, \eta \ra^k \leq 2  | \eta^2 + y^2 + x + \ii|, \]
and it is easy to estimate the derivatives of $(\eta^2 + y^2 + x + \ii)^{-1}$. According to the calculus of pseudodifferential operators, $L_j$ becomes a pseudodifferentail operator with symbol in
$$S^0(\R^4; \la y, \eta \ra ^{-k} \la x - \eta \ra ^{-j (1+\delta)} \la y + \xi \ra ^{-j({1\over 2} + \delta)}),$$
and the trace norm (resp. Hilbert-Schmidt norm) of $L_2$ (resp. $L_1$) can be estimated (see for instance, Proposition 9.2 and Theorem 9.4 in \cite{DS}) by
\begin{equation}\label{eq:6.2}
\|L_1\|_{{\rm HS}}^2 +\|L_2\|_{{\rm tr}}
 \leq C_0 \iiiint \la y, \eta \ra ^{-2k} \la x - \eta \ra^{-2 - 2\delta} \la y + \xi \ra^{-1 -2\delta}   dx d\xi dy d\eta 
 \end{equation}
 $$ \leq C_0' \iint \la y, \eta \ra^{-2k} dy d\eta \leq C_0''. $$
To deal with $M_j,\: j = 1,2,$ we will show that $(I - A) \tilde{k}_2^{\omega}$ is trace class operator and $(I - A) \tilde{k}_1^{\omega}$ is Hilbert-Schmidt one. 

 Notice that on the support of the symbol of  $(I - A)$ we have
\[ \la y, \eta \ra^k \geq |\eta^2 + y^2 + x +\ii| \,. \]
Taking into account the estimate $\partial_x^l\partial_y^m k_j(x,y)={\mathcal O}_{l, m}(\la x\ra^{-j(1+\delta)}\la y\ra^{-j({1\over 2}+\delta)})$, we get
\begin{equation}\label{eq:6.3}
 \|(I - A)k_1^{\omega} \|_{{\rm HS}}^2\: +\:\|(I - A)k_2^{\omega} \|_{{\rm tr}} 
 \leq
 C_1 \iiiint_{\la y, \eta \ra^k \geq |\eta^2 + y^2 + x +\ii|} \la x - \eta \ra^{-2 - 2\delta} \la y +  \xi \ra ^{-1 -2\delta} dx d\xi dy d\eta  \end{equation}
\[ \leq C_2 \iiint_{\la y, \eta \ra^k \geq |\eta^2 + y^2 + x +\ii|} \la x - \eta \ra^{-2 -2 \delta} dx dy d\eta \, 
 \leq C_2 \iiint_{\la y, \eta \ra^k \geq |\eta^2 + y^2 + \eta + u +\ii|} \la u \ra^{-2 -2 \delta} dudy d\eta \, \]
\[ \leq C_2'\iiint_{\la y, \eta \ra^k \geq |\eta^2 + y^2 + \eta + u|,\: \atop |u| \leq \frac{1}{2}\la y, \eta \ra^k} \la u \ra^{-2 -2\delta} dudy d\eta \, 
 + C_2'\iiint_{\la y, \eta \ra^k \geq |\eta^2 + y^2 + \eta + u|,\: \atop |u| \geq \frac{1}{2}\la y, \eta \ra^k} \la u \ra^{-2 -2 \delta} dudy d\eta \, \]
\[ \leq C_2' \Bigl(\iiint_{|u| \leq C_3, |y| \leq C_3, |\eta| \leq C_3} \la u \ra^{-2 -2 \delta} du dy d\eta + \iiint_{|u| \geq \frac{1}{2} \la y, \eta \ra^k} \la u \ra^{-2 -2 \delta} du dy d\eta \Bigr)\, \]
\[ \leq C_4 + C_5 \int \la u \ra ^{-2 -2 \delta}\Bigl(\int_0^{(2|u|)^{\frac{1}{k}}} r dr \Bigr) du
 \leq C_4  + C_6 \int \la u \ra^{-2 -2 \delta+ 2/k} du \leq  C_7 \,, \]
since $-2 -2 \delta + 2/k < -1.$ 

Using (\ref{eq:6.1}),(\ref{eq:6.2}), (\ref{eq:6.3}) 
 and the fact that $M$ is trace class (resp. Hilbert-Schmidt) operator if and only if $M^*$ is trace class (resp. Hilbert-Schmidt) operator,
we complete the proof of the lemma.
\end{proof}

{\footnotesize

\end{document}